\documentclass{IEEEtran}

\usepackage{graphicx}
\usepackage{mathrsfs}
\usepackage{amsthm,amsmath,amssymb,amsfonts,mathtools}
\usepackage{epstopdf}
\usepackage{flushend}
\usepackage{colortbl}
\usepackage[table]{xcolor}
\usepackage{tabularx}
\usepackage{makecell}
\usepackage{arydshln}
\usepackage{cite}
\usepackage{algorithmic}
\usepackage{algorithm}

\usepackage{float}
\usepackage{needspace}
\allowdisplaybreaks[4]

\usepackage[colorlinks = true,
            linkcolor = blue,
            urlcolor  = blue,
            citecolor = blue,
            anchorcolor = blue]{hyperref}

\newtheorem{definition}{Definition}

\newtheorem{lemma}{Lemma}
\newtheorem{remark}{Remark}
\newtheorem{theorem}{Theorem}

\usepackage{comment}
\newboolean{showcomments}
\setboolean{showcomments}{true}

\newcommand{\fliu}[1]{\ifthenelse{\boolean{showcomments}}
        { \textcolor{red}{(FL:  #1)}}{}}
\newcommand{\xpeng}[1]{\ifthenelse{\boolean{showcomments}}
        { \textcolor{blue}{(XP:  #1)}}{}}
 \newcommand{\xru}[1]{\ifthenelse{\boolean{showcomments}}
        { \textcolor{violet}{(XR:  #1)}}{}}

\renewcommand{\baselinestretch}{0.95}

\begin{document}

\title{Matrix-Valued Passivity Indices: Foundations, Properties, and Stability Implications}

\author{Xi Ru,
~Xiaoyu Peng,
~Xinghua Chen,
~Zhaojian Wang,
~Peng Yang,
~Feng Liu,~\IEEEmembership{Senior Member,~IEEE} 
}


\maketitle

\begin{abstract}
The passivity index, a quantitative measure of a system's passivity deficiency or excess, has been widely used in stability analysis and control. Existing studies mostly rely on scalar forms of indices, which are restrictive for multi-input, multi-output (MIMO) systems. This paper extends the classical scalar indices to a systematic matrix-valued framework, referred to as passivity matrices. A broad range of classical results in passivity theory can be naturally generalized in this framework. We first show that, under the matrix representation, passivity indices essentially correspond to the curvature of the dissipativity functional under a second-variation interpretation. This result reveals that the intrinsic geometric structure of passivity consists of its directions and intensities, which a scalar index cannot fully capture. For linear time-invariant (LTI) systems, we examine the structural properties of passivity matrices with respect to the Loewner partial order and propose two principled criteria for selecting representative matrices. Compared with conventional scalar indices, the matrix-valued indices capture the passivity coupling among different input–output channels in MIMO systems and provide a more comprehensive description of system passivity. This richer information leads to lower passivation effort and less conservative stability assessment.
\end{abstract}


%

\section{Introduction}


Dissipativity theory, first formulated in the seminal works of Willems~\cite{willems1972dissipative, willems1972dissipative2}, provides a foundational framework for the analysis and synthesis of complex dynamical systems. It characterizes system behavior through energy exchange with the environment using supply rates and storage functions, and offers a unified viewpoint for studying properties such as stability and passivity~\cite{hill2003stability}. Among the various forms of dissipativity, passivity is particularly significant due to its intimate connection with Lyapunov stability and its distinct structural property: the parallel and feedback interconnections of passive systems remain passive~\cite{van2000l2}.

This inherent robustness has made passivity-based control (PBC) an important paradigm in many engineering domains. In robotics, it enables safe physical interaction between humans and robots~\cite{hannaford2002time}. In power-electronics-dominated power systems, passivity plays a crucial role in analyzing the stability of grid-forming inverters and in preventing high-frequency resonance in renewable energy integration~\cite{han2022passivity, he2024passivity}. In networked control and cyber-physical systems, passivity is also a primary tool for counteracting time delays and packet losses, and for ensuring consensus or synchronization of multi-agent networks~\cite{chopra2006passivity, rajchakit2021robust}.

While the classical binary classification of a system as passive or non-passive provides a sufficient condition for stability, it is often too coarse for precise performance analysis. In practical applications, it is essential to quantify the degree of passivity a system possesses or how far it deviates from being passive. To address this need, passivity indices were introduced, commonly defined as input-feedforward passivity (IFP) and output-feedback passivity (OFP) indices~\cite{bao2007process}.

These indices have played a pivotal role in refining stability criteria. They allow for the stability analysis of feedback interconnections in which one subsystem may exhibit a shortage of passivity while the other provides an excess~\cite{zhu2014passivity}. Beyond simple feedback loops, passivity indices extend naturally to cascade interconnections via the secant criterion~\cite{yu2010passivity} and to symmetrically interconnected distributed systems, thereby enabling stability characterization for both linear and nonlinear dynamics under mixed feedforward and feedback structures~\cite{wu2011passivity}. This compensation principle often yields less conservative conditions than small-gain methods and supports the design of robust controllers for coupled systems~\cite{zakeri2022passivity}. The framework has also been generalized to broader classes of systems, including port-Hamiltonian systems~\cite{monshizadeh2019conditions} and discrete-time implementations~\cite{zhao2016feedback}, and has become indispensable for the passivation of non-passive systems via feedback or feedforward compensation~\cite{bao2007process, zhu2016passivity}.

Passivity indices have further proven effective in cooperative control, where they facilitate the analysis of heterogeneous multi-agent systems and ensure consensus under diffusive couplings~\cite{li2019consensus}. They also play an increasingly important role in modern power systems. At the device level, passivity indices guide the systematic passivation of power-electronic converters, such as current-controlled grid-connected VSCs~\cite{hans2018passivation}. At the network level, they support distributed stability assessment that captures heterogeneous and nonlinear bus dynamics in large-scale power systems~\cite{yang2019distributed}. Collectively, these developments underscore the importance of passivity indices as a rigorous basis for analyzing and enhancing the stability of interconnected and networked systems.

Despite these advances, most existing work is built on scalar passivity indices. In the standard formulation, the passivity level of a multi-input multi-output (MIMO) system is reduced to a single scalar value. In linear systems, this value corresponds to the minimum eigenvalue of the symmetric part of the transfer function; in nonlinear systems, it corresponds to the smallest sector bound~\cite{khalil2002nonlinear}. Recent studies show that such scalar indices cannot capture the directional heterogeneity of MIMO systems. A system may exhibit strong passivity in some input–output channels while remaining fragile in others due to cross-coupling effects. Collapsing this structure into the weakest direction discards essential information. 

This loss of structural information has two main consequences. First, it leads to conservative stability assessments. Scalar conditions implicitly assume isotropic behavior, and as observed in~\cite{chen2024limitations}, this assumption may indicate instability for interconnected grid-inverter systems that are in fact stable. Similar conclusions were drawn in~\cite{liceaga2015mimo}, which showed that scalar passivity indices fail to guarantee robustness in MIMO settings because directional fragility remains hidden. Second, it results in inefficient passivation strategies. Designs based on scalar measures, such as those in~\cite{zhu2014passivity, bao2007process}, often employ uniform compensation gains dictated by the worst-case direction, which imposes unnecessary control effort on channels that are already adequately passive. 

Recent work has also sought to broaden passivity indices beyond classical measures by reshaping the frequency-domain metric~\cite{chen2024extended}. They introduce a frequency-weighting rotation $R(\omega)$ and define extended IFP/OFP margins from the Hermitian part of $R^H(\omega)G(j\omega)$, enabling tailored low-frequency MIMO dynamics specifications. However, these margins remain essentially scalar and frequency-by-frequency, and they do not provide a systematic, direction-resolving matrix index for targeted passivation and less conservative interconnection analysis.

Although matrix-valued sector bounds have long been recognized in classical input-output stability theory \cite{safonov1980stability}, and QSR-dissipativity provides a unifying language for quadratic supply rates, a systematic framework for matrix-valued passivity indices remains undeveloped. It is crucial to distinguish the proposed framework from general QSR analysis. While QSR-dissipativity typically employs fixed weighting matrices to verify stability or performance in a binary manner, matrix-valued passivity indices serve as a quantitative metric of energy excess or shortage relative to the canonical passivity supply rate. However, existing methods lack the capability to explicitly quantify this property along specific directions, and therefore cannot capture the directional structure of MIMO systems. By reformulating these indices as intrinsic matrix-valued properties to be identified, we enable anisotropic passivation thereby overcoming the conservatism of scalar indices and the opacity of general QSR parameters. These motivations lead to the matrix-valued formulation presented in this paper.

The main contributions of this paper are threefold:
\begin{itemize}
    \item \textbf{Matrix-valued generalization of passivity indices.}
    We develop a systematic matrix-valued generalization of classical scalar passivity indices and show that standard passivity results, including interconnection properties, passivation control, and stability analysis, extend cleanly under this framework. The matrix-valued formulation yields improved analytical performance. It reduces the passivation effort by compensating only for deficient passivity directions and yields less conservative stability assessments by capturing the intrinsic multidirectional passivity characteristics of MIMO systems.

    \item \textbf{Intrinsic interpretation of passivity matrices.}
    We provide a geometric interpretation of passivity matrices in terms of the curvature of the dissipative functional, showing that they arise naturally from dissipativity theory. Passivity matrices offer a more faithful characterization of dissipative behavior. Their eigenvectors identify the most critical directions of energy exchange relevant to stability, while the corresponding eigenvalues quantify the energy surplus or deficit along each direction, thereby revealing a rich directional and intensity structure that scalar indices cannot capture. More interestingly, this finding might suggest a possibility of geometricizing the energy structure of dynamical systems.  

    \item \textbf{Structural analysis and computable representatives.}
    For linear time-invariant systems, partial-order analysis reveals that no single matrix can represent all passivity characteristics and that multiple incomparable candidates may coexist. We further propose two principled selection criteria that integrate effectively with LMI-based computation: maximizing the trace and maximizing the minimum eigenvalue.
\end{itemize}


The remainder of this paper is organized as follows. Section~\ref{sec-preliminary} reviews basic concepts from dissipativity theory and introduces the notation used throughout the paper. Section~\ref{sec-matrixindex} presents the matrix-valued formulation of passivity indices and explains its relationship with classical scalar definitions. Section~\ref{sec-physical interpretation} investigates the geometric interpretation and structural properties of passivity matrices, including their partial-order characteristics and selection criteria. Section~\ref{sec-stability} applies the proposed framework to interconnection analysis, passivation control, and stability assessment. Section~\ref{sec-examples} provides numerical examples that illustrate the analytical benefits of the matrix-valued formulation. Section~\ref{sec-conclusion} concludes the paper.

\section{Preliminaries}\label{sec-preliminary}

\subsection{Notations}
$\mathbb{S}^{n} = \{ A \in \mathbb{R}^{n \times n} \mid A_{ij} = A_{ji}, \forall i \neq j \}$ is the set of symmetric matrices. If $A\in\mathbb{S}^n$,  $\lambda_{\max}(A)$ and $\lambda_{\min}(A)$ denote the maximum and minimum eigenvalues of  $A$, respectively. For symmetric matrices $A$ and $B$, the notation $A\succeq B$ ($A\succ B$) denotes that $A-B$ is positive semidefinite (positive definite), and $A\preceq B$ ($A\prec B$) is defined analogously. The inertia of $A$ is the triple $(p,q,r)$, where $p$ (resp.$q$, $r$) is the number of positive (resp. negative, zero) eigenvalues of $A$, counted with multiplicities. For a general matrix $A$, its Hermitian symmetric part is denoted by $\bar{A}:= (A + A^H)/2$, where $A^H$ denotes the Hermitian transpose (conjugate transpose) of $A$. $\mathrm{diag}(A_1,\dots,A_k)$ denotes the block diagonal matrix with blocks $A_1,\dots,A_k$ on the diagonal. $L^{2}[0,T]$ denotes the Hilbert space of square–integrable functions on the interval $[0,T]$, with inner product 
$\langle f , g \rangle := \int_{0}^{T} f(t)^{\!*} g(t)\,dt$ and induced 
norm $\| f \| := \sqrt{\langle f , f\rangle}$. 

\subsection{Dissipative and Passive Systems}\label{sec-busmodel}


This subsection reviews the classical definitions of dissipativity and passivity. Consider a dynamical system represented by a state space model:
\begin{equation}\label{eq-dynamic}
    \begin{cases}
        \dot{x}=f(x,u)\\
        y=h(x,u)
    \end{cases}
\end{equation}
where state $x\in\mathcal{X}\subset\mathbb{R}^n$, input $u\in \mathcal{U}\subset\mathbb{R}^m$, and output $y\in\mathcal{Y}\subset\mathbb{R}^q$. The map $f:\mathcal{X}\times \mathcal{U}\rightarrow\mathbb{R}^n$ is locally Lipschitz, and $h:\mathcal{X}\times \mathcal{U}\rightarrow\mathcal{Y}$ is continuous. For each initial state $x_0 \in \mathcal{X}$, the system~\eqref{eq-dynamic} induces an operator $\mathbb{H}$ that maps any input signal $u(t)$ to the corresponding output signal $y(t)$. Throughout this paper, we consider MIMO systems with matched input-output dimensions, i.e., $m=q$ in~\eqref{eq-dynamic}. Without loss of generality, we assume that the origin is an equilibrium of~\eqref{eq-dynamic}. That is, $f(0,0)=0$, $h(0,0)=0$.


\begin{definition}[Dissipativity\cite{van2000l2}]
    The system~\eqref{eq-dynamic} is called dissipative with respect to a supply rate $s(u(t),y(t))$ if there exists a positive semidefinite (p.s.d.) function $V(x)$, referred to as the storage function, s.t. for all initial conditions $x(t_0)=x_0\in\mathcal{X}$ and for all admissible inputs $u$ and all $t_1\ge t_0$ the following inequality holds
    \begin{equation}
        V(x(t_1))\le V(x(t_0))+\int_{t_0}^{t_1}s(u(t),y(t))\, \mathrm{d} t
        \label{eq-dissipativity}
    \end{equation}
\end{definition}

\begin{definition}[Passivity\cite{khalil2002nonlinear}] 
    The system~\eqref{eq-dynamic} with $m=q$ is passive if the supply rate $s(t)$ in~\eqref{eq-dissipativity} is
    \begin{equation}
        s(t)=u^{\top}y
    \end{equation}
    If the storage function $V$ is smooth, then~\eqref{eq-dissipativity} implies
    \begin{equation}
        \dot{V}=\frac{\partial V}{\partial x}f(x,u)\le u^{\top}y, \forall (x,u)\in\mathcal{X}\times\mathcal{U}
    \end{equation}
    Moreover, the system is input-feedforward passive if $\dot{V} \le u^{\top}y - u^{\top}\delta(u)$ for some function $\delta$, 
    and input strictly passive if $\dot{V} \le u^{\top}y - u^{\top}\delta(u)$ with $u^{\top}\delta(u) > 0$, $\forall u \ne 0$.
    Similarly, it is output-feedback passive if $\dot{V} \le u^{\top}y - y^{\top}\rho(y)$ for some function $\rho$, 
    and output strictly passive if $\dot{V} \le u^{\top}y - y^{\top}\rho(y)$ with $y^{\top}\rho(y) > 0$, $\forall y \ne 0$.
    In addition, it is strictly passive if $\dot{V} \le u^{\top}y - \psi(x)$ for some positive definite function $\psi$.
    \label{de-dynamical passive}
\end{definition}

\begin{definition}[Zero-state observability~\cite{khalil2002nonlinear}]
The system~\eqref{eq-dynamic} is zero-state observable if
no solution of $\dot{x}=f(x,0)$ can stay identically in
$S=\{\,x\in\mathbb{R}^n \mid h(x,0)=0\,\}$ other than the trivial
solution $x(t)\equiv 0$.
\end{definition}

\subsection{Conventional scalar Passivity Indices}
In Definition~\ref{de-dynamical passive}, assume the dissipation terms take the linear forms $\delta(u)=\phi u$ and $\rho(y)=\xi y$ with $\phi, \xi\in\mathbb{R}$. The scalars $\phi$ and $\xi$ are called the input-feedforward and the output-feedback passivity indices, respectively. A positive index indicates passivity excess, whereas a negative index indicates a passivity shortage. The following definition consolidates these notions into a single scalar formulation.

\begin{definition}[scalar passivity indices\cite{bao2007process,zhu2014passivity}]
    The system~\eqref{eq-dynamic} is input-feedforward output-feedback passive (IF-OFP) if there exists a continuously differentiable p.s.d. storage function $V(x)$ s.t. 
    \begin{equation}
        \dot{V}\le u^{\top}y-\phi u^{\top}u-\xi y^{\top}y, \forall(x,u)\in\mathcal{X}\times\mathcal{U}
        \label{eq-passivity indices}
    \end{equation}
    for some $\phi,\xi\in\mathbb{R}$. Moreover, the system is called $(\phi,\xi)$--passive, where $\phi,\xi\in\mathbb{R}$ are the input-feedforward passivity index and the output-feedback passivity index, respectively.
    \label{de-scalar passivity indices}
\end{definition}

Before introducing matrix-valued passivity indices, we recall the classical notion of sector, which helps interpret passivity in terms of directional input–output relations.

\begin{definition}[Sector\cite{khalil2002nonlinear}]
    A memoryless function $h:[0,\infty)\times\mathbb{R}^m\rightarrow\mathbb{R}^m$ is said to belong to the sector:
    \begin{itemize}
        \item $[K_1,\infty]$ if $u^{\top}[h(u)-K_1 u]\ge 0$
        \item $[0,K_2]$ if $h^{\top}(u)[h(u)-K_2 u]\le 0$
        \item $[K_1,K_2]$ with $K=K_2-K_1=K^{\top}\succ 0$ if $[h(u)-K_1 u]^{\top}[h(u)-K_2 u]\le 0$
    \end{itemize}
    In all cases, the inequality should hold for all $(t,u)$.
\end{definition}

\section{From Passivity Indices to Passivity Matrices}\label{sec-matrixindex}
\subsection{Matrix-Valued Indices of Classical passivity}

The scalar passivity indices introduced earlier correspond to one-dimensional sector bounds. We now generalize this idea by replacing scalar dissipation terms with symmetric matrices, yielding the notion of matrix-valued passivity indices, or passivity matrices for short.

\begin{definition}[Matrix-valued passivity indices]
    The system~\eqref{eq-dynamic} is input-feedforward output-feedback passive (IF-OFP) if there exists a continuously differentiable p.s.d. storage function $V(x)$ s.t. 
    \begin{equation}
        \dot{V}\le u^{\top}y- u^{\top}\Phi u-y^{\top}\Xi y, \forall(x,u)\in\mathcal{X}\times\mathcal{U}
        \label{eq-passivity matrix}
    \end{equation}
    Moreover, the system is called $(\Phi,\Xi)$--passive, where the matrices $\Phi,\Xi\in\mathbb{S}^{m}$ are called the input-feedforward passivity matrix (IFPM) and the output-feedback passivity matrix (OFPM), respectively.
    \label{de-matrix passivity indices}
\end{definition}

 For a better understanding of the properties of matrix-valued passivity indices, we present the following observations.
\begin{itemize}
    \item[a)] 
    The matrix-valued passivity indices are natural generalizations of their scalar counterparts. The former reduces to the latter for SISO systems. For MIMO systems, the latter can serve as a measure of the former, as we explain later in this part. 
    \item[b)]
    If $\Phi = 0$ and $\Xi \ne 0$, the system is output-feedback passive with the OFPM $\Xi$, denoted by $OFP(\Xi)$. Furthermore, $\Xi\succ 0$ implies output strictly passive ($OSP(\Xi)$), indicating that the system possesses passivity excess in all output directions. Conversely, $\Xi \prec 0$ implies global shortage of passivity.
    \item [c)] If $\Xi = 0$ and $\Phi \ne 0$, the system is input-feedforward passive with the IFPM $\Phi$, denoted by $IFP(\Phi)$. Similarly, \(\Phi \succ 0\) yields input strict passivity (\(ISP(\Phi)\)), corresponding to passivity excess in all input directions, whereas \(\Phi \prec 0\) indicates a global passivity shortage.
\end{itemize}

When the passivity matrices $\Xi$ or $\Phi$ are neither positive definite nor negative definite, the system is neither globally strictly passive nor globally passive-deficient. Instead, it exhibits \emph{direction-dependent passivity}. Along eigen-directions associated with positive eigenvalues, the system exhibits passivity excess; along zero-eigenvalue directions, it behaves losslessly; and along directions corresponding to negative eigenvalues, it exhibits a passivity shortage. This mixed passivity profile is intrinsic to MIMO systems and cannot be captured by scalar passivity indices.

We next present a basic structural property of passivity matrices that is essential for their interpretation and computation. 


\begin{lemma}
    If the system~\eqref{eq-dynamic} is output--feedback passive with an output--feedback map $\rho(y)$, then the system possesses the output--feedback passivity matrix (OFPM) $\Xi$ provided that $\rho(y)$ belongs to the sector $[\Xi, \infty]$. 
    Similarly, if the system~\eqref{eq-dynamic} is input--feedforward passive with an input--feedforward map $\delta(u)$, then the system possesses the input--feedforward passivity matrix (IFPM) $\Phi$ provided that $\delta(u)$ belongs to the sector $[\Phi, \infty]$.
    \label{le-order}
\end{lemma}
\begin{proof}
    The proof can be found in Appendix~\ref{sec-app-orderstructure}.
\end{proof}

\begin{remark}
    If the system possesses an OFPM $\Xi$, then any matrix $M \preceq \Xi$ is also an OFPM of the system. 
    Likewise, if the system possesses an IFPM $\Phi$, then any matrix $N \preceq \Phi$ is also an IFPM of the system. 
    In general, the passivity matrix that is maximal under the Loewner partial order provides the least conservative and most informative quantification of the system's passivity.
    \label{re-maxorder}
\end{remark}
By choosing $\xi = \lambda_{\min}(\Xi)$ and $\phi = \lambda_{\min}(\Phi)$, 
the matrix-valued passivity indices reduce to their scalar counterparts. 
Thus, the matrix formulation does not alter the existence of passivity indices; instead, it refines their descriptive capability. 
Specifically, the passivity matrices capture the passivity characteristics of different input-output channels in a MIMO system, whereas the scalar indices reflect only the weakest passivity direction. 
Such a difference in granularity is also reflected in stability analysis: 
Since the passivity matrices encode multidirectional passivity characteristics, the stability conditions derived from them are typically less conservative than those based on scalar passivity indices. 
This distinction will be further illustrated in the subsequent section.

Definition~\ref{de-matrix passivity indices} is a structured instance of QSR dissipativity. Indeed, the dissipation inequality~\eqref{eq-passivity matrix} is equivalent to dissipativity with the quadratic supply rate
\begin{equation*}
    s(u,y)=\begin{bmatrix}
        u\\y
    \end{bmatrix}^\top\begin{bmatrix}
        -\Phi&\frac{1}{2}I\\\frac{1}{2}I&-\Xi
    \end{bmatrix}\begin{bmatrix}
        u\\y
    \end{bmatrix}
\end{equation*}
Our focus is on using this structure to obtain an interpretable, direction-dependent measure of passivity shortage or excess for MIMO systems, and to leverage it in interconnection and passivation results.

\subsection{Passivity Matrices for Linear Time-Invariant Systems}
We now extend the matrix-valued passivity indices to linear time-invariant (LTI) systems.

\begin{definition}[Matrix-valued IFP index]
    For a stable linear system $G(s)$, the frequency-dependent input-feedforward passivity matrix (IFPM) at frequency $\omega$ is defined as
    \begin{equation}
        H(\omega) \triangleq \frac{1}{2}\left(G(j\omega)+G^H(j\omega)\right)
        \label{eq-Hermitian transfer}
    \end{equation}
    A Hermitian matrix $\Phi$ is called an IFPM of $G(s)$ if $ \Phi \preceq H(\omega),\forall\omega\in\mathbb{R}$.
    \label{def-linear-IFPM}
\end{definition}

\begin{definition}[Matrix-valued OFP index]
    For a minimum phase linear system $G(s)$, the frequency-dependent output-feedback passivity matrix (OFPM) at frequency $\omega$ is given by
    \begin{equation}
        K(\omega) \triangleq \frac{1}{2}\left(G^{-1}(j\omega)+\left[G^{H}(j\omega)\right]^{-1}\right)
    \end{equation}
    A Hermitian matrix $\Xi$ is said to be the OFPM of $G(s)$ if $\Xi \preceq K(\omega),\forall\omega\in\mathbb{R}$.
    \label{def-linear-OFPM}
\end{definition}

Note that Definitions~\ref{def-linear-IFPM} and~\ref{def-linear-OFPM} are equivalent to Definition~\ref{de-matrix passivity indices} by invoking the KYP lemma~\cite{khalil2002nonlinear}. These definitions bridge the real positivity in the frequency domain and the passivity matrices in the state-space, which will serve as the basis for structural analysis in the next section.

\begin{remark}[Relation to frequency-weighted passivity indices]
A frequency-domain extension of passivity indices introduces a weighting matrix $R(\omega)$ and defines the margin
\begin{equation}
\phi=\tfrac12\lambda_{\min}\!\left(R^H(\omega)G(j\omega)+G^H(j\omega)R(\omega)\right)
\label{eq-frequencyindices}
\end{equation}
thereby incorporating frequency-weighted MIMO specifications~\cite{chen2024extended}.
In contrast, this paper extends indices from scalars to matrices to preserve directional structure.
These two extensions can be combined by defining the weighted IFPM
\begin{equation*}
H_R(\omega):=\tfrac12\!\left(R^H(\omega)G(j\omega)+G^H(j\omega)R(\omega)\right)
\end{equation*}
and applying our matrix-valued framework to $H_R(\omega)$; the margin in~\eqref{eq-frequencyindices} is then a scalar reduction via $\lambda_{\min}\!\left(H_R(\omega)\right)$.
\end{remark}


\section{Geometric Interpretation of Passivity Matrix}\label{sec-physical interpretation}
Matrix-valued passivity indices extend scalar measures by revealing how passivity is distributed across different input–output directions. This directional information offers a more refined view of system passivity and is particularly useful for stability analysis or controller design.

This section develops an intrinsic interpretation of these matrices within the dissipativity theory, showing how their directional structure arises naturally from the curvature of the dissipation functional. We then examine the structural properties of passivity matrices for LTI systems under the Loewner partial order and present two practical principles for computing representative matrices. 
\subsection{Intrinsic Geometric Structure}
Consider the nonlinear control system~\eqref{eq-dynamic}.
The supply rate is taken in the quadratic form of
\begin{equation}
    s(t) = 
\begin{bmatrix} u \\[1mm] y \end{bmatrix}^{\top}
\begin{bmatrix}
    Q_{uu} & Q_{uy} \\ Q_{uy}^\top & Q_{yy}
\end{bmatrix}
\begin{bmatrix} u \\[1mm] y \end{bmatrix}
\label{eq-dissipation rate}
\end{equation}
where the block matrix is symmetric and characterizes the dissipativity structure.
For a continuously differentiable nonnegative storage function $V:\mathcal{X}\to\mathbb{R}_{\ge 0}$, the instantaneous dissipation is defined as 
\begin{equation}
    d(t) := s(t) - \dot{V}=s(t)-\nabla V(x)^{\top} f(x,u).
    \label{eq-instantaneous dissipation}
\end{equation}
Equation~\eqref{eq-instantaneous dissipation} reflects the instantaneous energy balance at  $t$: the supplied power minus the rate of change of stored energy equals the instantaneous rate at which the system dissipates energy. However, instantaneous dissipation does not describe how the system dissipates energy along an entire input–state–output trajectory. To capture the cumulative effect over a finite horizon and obtain a scale-independent measure, we consider the time-averaged dissipation functional:
\begin{equation}
\begin{aligned}
    J_T[u] :=& \frac{1}{T}\int_0^T d(t)\, \mathrm{d} t\\
\end{aligned}
\end{equation}
This construction provides a meaningful measure of the system’s average dissipative behavior over the interval $[0,T]$. 

To understand how the system dissipates energy along different input modes, it is necessary to examine how $J_T$ changes under infinitesimal perturbations of the input. The dissipativity directions emerge naturally from the second variation of the dissipation functional. Mathematically, whenever a functional is twice Fréchet differentiable on a Hilbert space, its second variation induces a continuous symmetric bilinear form, which corresponds to a bounded self-adjoint operator~\cite{reed1980methods}.
The spectral theorem for bounded self-adjoint operators provides a spectral decomposition of this operator, yielding orthogonal spectral components that describe how the functional bends along different perturbation directions.

To evaluate the second variation of $J_T$, we consider perturbations around a closed-loop trajectory $(x^\ast(t),u^\ast(t),y^\ast(t))$. Linearizing the system dynamics yields
\[
\delta\dot{x} = A(t)\delta x + B(t)\delta u, 
\qquad
\delta y = C(t)\delta x + D(t)\delta u
\]
where the matrices $(A,B,C,D)$ are the Jacobians of $(f,h)$ along the trajectory. Assume these matrices are measurable and essentially bounded on $[0,T]$.  
Solving the variational dynamics with zero initial deviation yields:
\begin{equation}
    \begin{cases}
        \delta x(t) = \int_0^T \Psi(t,\tau) B(\tau)\,\delta u(\tau)\, \mathrm{d} \tau, \qquad\\
        \delta y(t) = \int_0^T C(t)\Psi(t,\tau)B(\tau)\,\delta u(\tau)\, \mathrm{d} \tau + D(t)\delta u(t)
    \end{cases}
    \label{eq-disturbed system}
\end{equation}
where $\Psi(t,\tau)$ is the state transition matrix.

Since the storage function is not unique, we consider periodic trajectories $x(0)=x(T)$ and periodic perturbations $\delta x(0)=\delta x(T)$ so that the storage-related boundary terms cancel. This isolates the intrinsic input–output dissipation. When the second variation is computed under this boundary condition, substituting~\eqref{eq-disturbed system} into the quadratic supply rate leads to the second
variation of the dissipation functional:
\begin{equation*}
    \begin{aligned}
        \delta^2 J_T[u] =& \frac{1}{T}\int_0^T[\delta^2s(t)]\, \mathrm{d} t\\
        =&\frac{1}{T}\iint_{0}^{T}\!\delta u(t)^{\top} K_Q(t,\tau) \delta u(\tau)\, \mathrm{d} \tau\, \mathrm{d} t,
    \end{aligned}
\end{equation*}
where the kernel $K_Q(t,\tau)$ is given by
\begin{equation}
\begin{aligned}
    K_Q(t,\tau) =& Q_{uu}\delta(t-\tau) + Q_{uy} G(t,\tau)^{\top} 
    + G(t,\tau) Q_{uy} \\
    &+ \int_0^T G(s,t)^{\top} Q_{yy} G(s,\tau)\, \mathrm{d} s
    \label{eq-kernel}
\end{aligned}
\end{equation}
and 
\[
G(t,\tau) := C(t)\Psi(t,\tau)B(\tau) + D(t)\delta(t-\tau)
\]

This defines a bounded self-adjoint operator
\begin{equation}
    \bigl(\mathbb{D}^{(Q)}_{T}\delta u\bigr)(t)
    := \int_{0}^{T} K_Q(t,\tau)\,\delta u(\tau)\, \mathrm{d} \tau
    \label{eq-DToperator}
\end{equation}


By the spectral theorem for bounded self-adjoint operators, there exists a projection-valued spectral measure $E(\cdot)$ s.t.
\begin{equation}
    \mathbb{D}^{(Q)}_{T} = \int_{\sigma(D_T^{(Q)})} \lambda\, \mathrm{d}E(\lambda)
\end{equation}
Here $\sigma(\cdot)$ denotes the spectrum of a bounded self-adjoint operator. Since $\mathbb{D}^{(Q)}_{T}$ is self-adjoint, its spectrum is real, i.e., $\sigma(\mathbb{D}^{(Q)}_{T})\subset\mathbb{R}$.
For any input perturbation $\delta u\in L_2([0,T],\mathbb{R}^m)$, the second variation admits the spectral representation
\begin{equation}
    \delta^2 J_T[\delta u]
    =
    \frac{1}{T}\big\langle \mathbb{D}^{(Q)}_{T}\delta u,\delta u\big\rangle
    =
    \frac{1}{T}\int_{\sigma(\mathbb{D}^{(Q)}_{T})} \lambda\, \mathrm{d}\big\langle E(\lambda)\delta u,\delta u\big\rangle
    \label{eq:spectral_repr_second_variation}
\end{equation}
Here $E(\lambda)\coloneqq E((-\infty,\lambda])$ denotes the right-continuous spectral family associated with $E(\cdot)$, and the integral is understood in the Stieltjes sense.
Therefore, the spectral decomposition of $\mathbb D_T^{(Q)}$ reveals the curvature structure of $J_T$. The spectral projections $E(\cdot)$ provide a spectral resolution of $\mathbb{D}^{(Q)}_{T}$, yielding mutually orthogonal spectral components of any perturbation $\delta u$.
We interpret these spectral subspaces as \emph{\textbf{generalized dissipativity directions}}: a perturbation $\delta u$ contributes to $\delta^2 J_T$ through its orthogonal components on these subspaces, and each component is weighted by the corresponding spectral value $\lambda$. The quantity  $\left\langle \mathbb D_T^{(Q)}\delta u,\delta u \right\rangle$ aggregates the dissipativity curvature across the spectral components of $\delta u$, while the normalized Rayleigh quotient $\left\langle \mathbb D_T^{(Q)}\delta u,\delta u \right\rangle/\left\|\delta u\right\|_2^2$ quantifies the associated average dissipativity intensity of $\delta u$. Positive contributions correspond to increasing energy absorption and negative contributions correspond to energy release. If $\mathbb D_T^{(Q)}$ is compact, the above spectral decomposition reduces to a discrete eigenfunction expansion $\mathbb{D}^{(Q)}_{T}\phi_k = \lambda_k \phi_k$ with an orthonormal eigenbasis $\{\phi_k\}$, choosing $\delta u=\phi_k$ recovers the discrete modal interpretation with $\delta^2 J_T[\phi_k]=\lambda_k/T$.

\begin{remark}
The concept of dissipativity directions does not rely on the supply rate being globally quadratic. 
If the supply rate is twice Fréchet differentiable in $(u,y)$ along a given trajectory and the induced second variation defines a continuous symmetric bilinear form on $L^2([0, T],\mathbb{R}^m)$, then it admits a Riesz representation by a bounded self-adjoint operator. The above spectral analysis of dissipativity directions still applies. Compactness and hence a discrete eigenfunction expansion requires additional regularity assumptions.  
\label{re-geoqsrsupplyrate}
\end{remark}

In this work, we restrict attention to quadratic supply rates in order to obtain a closed-form representation of the corresponding operator in terms of constant matrices $Q_{uu}$, $Q_{uy}$, and $Q_{yy}$.
Passivity directions arise as a special case of this general construction.  
For the standard passivity supply rate $s(u,y)=u^\top y$, one has $Q_{uy}=\tfrac{1}{2}I$ and $Q_{uu}=Q_{yy}=0$, and the operator~\eqref{eq-DToperator} reduces to
\begin{equation}
    \bigl(\mathbb{D}^{(Q)}_{T}\delta u\bigr)(t)
    = \frac{1}{2}\int_{0}^{T} \!\bigl[G(t,\tau)+G(t,\tau)^{\!\top}\bigr]\delta u(\tau)\, \mathrm{d} \tau
\end{equation}
whose spectral subspaces (and eigenfunctions in the compact case) therefore capture the characteristic passivity directions of the system.

For the passivity of linear time–invariant systems, the kernel of the
dissipativity operator $\mathbb{D}_{T}^{(Q)}$ reduces to
\begin{equation}
    K_Q(t,\tau)=\tfrac12\big(g(t-\tau)+g(\tau-t)^\top\big)
\end{equation}
with $g(\cdot)$ denoting the impulse response of the system. The passivity
directions and intensities are therefore determined by the eigenfunction–eigenvalue problem
\begin{equation}
    \int_{0}^{T}K_Q(t,\tau)\phi_k(\tau)\, \mathrm{d} \tau
= \lambda_k\,\phi_k(t)
\end{equation}
Since the kernel $K_Q(t,\tau)$ of an LTI system depends only on the time
difference $t-\tau$, the dissipativity operator $\mathbb{D}_{T}^{(Q)}$ is
a convolution–type self–adjoint operator. Such operators commute with
time–shift operators and, when considered on an infinite time horizon
(or in the asymptotic limit $T\to\infty$), are diagonalized by the
Fourier basis. As a result, the eigenfunction–eigenvalue structure
asymptotically takes the form
\begin{align}
u_{i,\omega}(t)&=e^{j\omega t}v_i(\omega)\\
\mathbb{D}_{\infty}^{(Q)}u_{i,\omega}(t)&=\lambda_i(\omega)\,u_{i,\omega}(t)
\end{align}
where $v_i(\omega)$ is an eigenvector of the Hermitian matrix $H(\omega)$
defined as~\eqref{eq-Hermitian transfer} and $\lambda_i(\omega)$ is the
corresponding eigenvalue. Thus, in the frequency domain, the passivity
operator $\mathbb{D}_{\infty}^{(Q)}$ reduces to multiplication by the
Hermitian part of the transfer function, which corresponds exactly to
the frequency–dependent IFPM, while $\mathbb{D}_{T}^{(Q)}$ on finite horizons provides a finite–time approximation of this spectral
characterization. The eigenvalues $\lambda_i(\omega)$ quantify the
passivity intensity at each frequency, while the associated vectors
$v_i(\omega)$ represent the passivity directions in the input space. If
$G(s)$ is a minimum phase system, then $G^{-1}(s)$ exists and is stable,
thus the OFPM of $G(s)$ corresponds to the IFPM of the inverse system
$G^{-1}(s)$. Consequently, the dissipativity geometry of LTI systems is trajectory–independent, and the passivity matrix fully characterizes both
the magnitude and the direction of energy dissipation at frequency $\omega$.

\begin{remark}
For nonlinear systems, the operator $\mathbb{D}_{T}^{(Q)}$ varies with the chosen trajectory $(x^\ast(t),u^\ast(t),y^\ast(t))$. Consequently, dissipativity (and passivity) directions are 
inherently trajectory–dependent, and there generally does not exist a single global direction that is valid for all operating points. 
A global passivity matrix must satisfy the 
dissipativity inequality~\eqref{eq-dissipativity} for all $(x,u)$, and thus corresponds to the lower envelope of all local dissipativity directions. 
\end{remark}

A matrix--valued passivity index contains information about both the passivity directions and the passivity intensities, whereas a scalar index reflects only the weakest passivity intensity.  The meaning of a matrix--valued passivity index can also be understood from the perspective of input--output coordinate transformations.
Consider a system~\eqref{eq-dynamic} with an IFPM $\Phi$. The symmetric IFPM can be orthogonally decomposed as $\Phi=Q_{in}^{\top}R_{in}Q_{in}$, where $Q_{in}$ is an orthogonal matrix and $R_{in}$ is a diagonal matrix whose diagonal entries are the eigenvalues of $\Phi$. According to the definition of IFP, the passivity inequality could be rewritten as:
\begin{equation}
\begin{aligned}
    \dot{V}&\le u^{\top}y-u^{\top}\Phi u\\
    &=(Q_{in}u)^{\top}(Q_{in}y)-(Q_{in}u)^{\top}R_{in}(Q_{in}u)
\end{aligned}
\end{equation}

This expression shows that the orthogonal decomposition corresponds to a coordinate transformation of the input and output signals,
namely $u' = Q_{in}u$ and $y' = Q_{in}y$, under which the IFPM $R_{in}$ becomes fully
diagonal. Each diagonal entry of $R_{in}$ represents the passivity intensity associated with a decoupled input--output subchannel, while the columns of $Q$ serve as the eigenvectors of $\Phi$ and characterize the directions of coordinate transformation in which the passivity property becomes completely decoupled.

The OFPM $\Xi$ admits the same orthogonal decomposition $\Xi=Q_{out}^{\top}R_{out}Q_{out}$ and associated input–output coordinate transformation:
\begin{equation}
    \dot{V}\le (Q_{out}u)^{\top}(Q_{out}y)-(Q_{out}y)^{\top}R_{out}(Q_{out}y)
\end{equation}
The coordinate transformation $u' = Q_{out}u$ and $y' = Q_{out}y$ leading to a fully decoupled representation of the output–feedback passivity channels.

The above decomposition also clarifies what is—and what is not—captured by constant matrix representatives under quadratic supply rates; we briefly comment on possible extensions beyond matrix-valued indices.

\begin{remark}
    More generally, the passivity matrices $\Phi$ and $\Xi$ in~\eqref{eq-passivity matrix} can be interpreted as second-order tensors that characterize the quadratic component of energy dissipation. This interpretation reveals a natural pathway for extending the proposed framework to strictly non-quadratic systems. While this paper focuses on quadratic supply rates, complex nonlinear damping phenomena may require higher-order tensor indices to accurately capture the passivity intensity. Thus, the matrix-valued formulation serves as the foundational second-order instance of a broader tensor-valued passivity theory for polynomial dynamical systems.
\end{remark}



\subsection{Order-Theoretic Properties of Passivity Matrices in LTI Systems}
Consider the linear time-invariant system\label{sec-orderpropwety}
\begin{equation}
    \begin{cases}
        \dot{x} = Ax + Bu\\
        y = Cx + Du
    \end{cases}
    \label{eq-linear system}
\end{equation}
with transfer function $G(s) = C(sI-A)^{-1}B + D$ and a quadratic storage function $V(x) = x^{\top}Px$.  
The IFP and OFP matrices of~\eqref{eq-linear system} can be computed from the following LMI:
\begin{equation}
    W\overset{\mathrm def}{=}
    \left[
    \begin{array}{c: c}
        \makecell{PA+A^{\top}P \\ + 2C^{\top}\Xi C}  
        & \makecell{PB - C^{\top} \\+ 2C^{\top}\Xi D} \\
        \noalign{\vspace{5pt}}
        \hdashline
        \noalign{\vspace{5pt}}
        \makecell{B^{\top}P-C \\+ 2D^{\top}\Xi C}  
        & \makecell{-(D+D^{\top}) \\+ 2\Phi+2D^{\top}\Xi D}
    \end{array}
    \right] \prec 0
    \label{eq-LMI}
\end{equation}
Setting $\Xi = 0$ yields the IFP case, while setting $\Phi = 0$ yields the OFP case. 

The passivity matrices definitions in Definition~\ref{def-linear-IFPM} and~\ref{def-linear-OFPM} can be interpreted using the Löwner partial order on symmetric matrices. 
For the IFP case, the matrices 
\[
\mathcal{H} \triangleq \{\, H(\omega) : \omega\in\mathbb{R} \,\} \subset \mathbb{S}_m
\]
form a frequency-indexed family of symmetric matrices.  
For the OFP case, the corresponding family is  
\[
\mathcal{K} \triangleq \{\, K(\omega) : \omega\in\mathbb{R} \,\} \subset \mathbb{S}_m
\]

\begin{definition}[Lower bound]
Let $\mathcal{A}\subset \mathbb{S}^m$ be a nonempty set of symmetric matrices.
A matrix $C\in \mathbb{S}^m$ is called a lower bound of $\mathcal{A}$
if
\[
C \preceq A, \qquad \forall\, A\in\mathcal{A}.
\]
The set of all lower bounds of $\mathcal{A}$ is denoted by
\[
\mathcal{L}(\mathcal{A}) 
\;=\; \{\, X\in \mathbb{S}^m : X \preceq A,\ \forall A\in\mathcal{A} \,\}
\]
\end{definition}

A symmetric matrix $\Phi$ is an IFPM of $G(s)$ if and only if it is a lower bound of $\mathcal{H}$. Similarly, the OFP index~$\Xi$ is a lower bound of the matrix family $\mathcal{K}$. The LMI in~\eqref{eq-LMI} provides exactly such a lower bound, in the sense that any solution $(P,\Phi,\Xi)$ to the LMI yields a matrix $\Phi\! \in \mathcal{L}(\mathcal{H})$ (or $\Xi \in \mathcal{L}(\mathcal{K})$). Motivated by Remark~\ref{re-maxorder}, we seek lower bounds that are maximal in the Löwner order, as such choices provide the least conservative representation of the system’s passivity.  This leads naturally to the notion of a maximal lower bound.

\begin{definition}[Maximal lower bound]
Let $\mathcal{A}\subset \mathbb{S}^m$ be a nonempty set of symmetric matrices.  
A matrix $C\in\mathcal{L}(\mathcal{A})$ is called a \emph{maximal lower bound} 
of $\mathcal{A}$ if it is a maximal element of the partially ordered set 
$(\mathcal{L}(\mathcal{A}),\preceq)$; that is, 
whenever $D\in\mathcal{L}(\mathcal{A})$ satisfies $C\preceq D$, then $D=C$. The set of all maximal lower bounds of $\mathcal{A}$ is denoted by $\mathcal{L}_{\max}(\mathcal{A})$.
\end{definition}

According to this definition, all $\Phi \in \mathcal{L}(\mathcal{H})$ and 
$\Xi \in \mathcal{L}(\mathcal{K})$ are valid candidates for passivity matrices 
with maximal order.

\begin{lemma}[Maximal lower bounds of two symmetric matrices~\cite{stott2016maximallowerboundslowner}]
    Let $A,B\in\mathbb{S}^n$, and let $(p,q,r)$ denote the inertia of $A-B$. Then the set of maximal lower bounds of $\left\{A,B\right\}$ is nonempty and there exists a bijection $\psi:\mathbb{M}_{p,q}\rightarrow\mathcal{L}_{\max}(\left\{A,B\right\})$.
    \label{le-maximallowerbounds}
\end{lemma}
This result characterizes the full family of maximal lower bounds of two 
symmetric matrices as a $pq$-dimensional manifold. 
Lemma~\ref{le-maximallowerbounds} and other fundamental results in matrix theory show that, in general, a set of symmetric matrices does not admit a unique greatest lower bound under 
the Löwner order unless all matrices in $\mathcal{A}$ are mutually comparable. 
For stable LTI systems, however, the set $\mathcal{L}(\mathcal{A})$ typically 
contains multiple maximal lower bounds, none of which dominate the others.

In passivity analysis, relying solely on maximality in the Löwner order is 
therefore not sufficient. Such a rule does not align with the LMI-based 
computation of passivity matrices, nor does it offer structural properties that 
help explain the underlying passivity behavior of the system. 
Additional selection principles are needed, built on top of the partial order, 
to extract passivity matrices that better reflect the geometry.

To refine the LMI-based characterization in~\eqref{eq-LMI}, we introduce two 
practically meaningful selection principles, each highlighting a different 
structural aspect of passivity matrices.

\subsubsection{\textbf{Maximize the trace}}
Since a matrix-valued passivity index captures both passivity directions and intensities, choosing the trace-maximal passivity matrix is natural. A larger trace represents stronger overall passivity when all directions are considered together.

This selection principle is also consistent with the geometric construction 
of~\cite{BURGETH2007277}, where the infimum of a finite family of positive semidefinite matrices is obtained as the matrix whose 
associated “matrix ball’’ has the smallest enclosing radius—equivalently, 
the lower bound with the largest trace. Based on this alignment, we approximate the same principle in the LTI passivity problem by maximizing the trace subject to the LMI in~\eqref{eq-LMI}. The resulting trace-maximal IFPM is denoted by $\Phi^{\mathrm{tr}}$, and the corresponding OFPM by $\Xi^{\mathrm{tr}}$.

A trade-off still exists. A trace-maximal matrix does not maximize its minimum eigenvalue, so the weakest passivity direction may become smaller. However, this reduction is structurally constrained: since the passivity matrix must remain a lower bound of the entire set $\mathcal{H}$, any degradation in the weakest passivity direction is necessarily bounded. In practice, this constraint yields a balanced outcome: the overall passivity level is increased as much as possible, while the sacrifice in the weakest direction remains within the limits imposed by the Loewner ordering. Thus, if we evaluate passivity over the entire input--output space, the trace-maximal choice provides a reasonable and comprehensive selection rule.

\subsubsection {\textbf{Maximize the minimum eigenvalue}}
Before introducing the second principle, we recall the notion of a tangency 
constraint, which offers a useful structural perspective on the lower bounds of 
symmetric matrices~\cite{stott2016maximallowerboundslowner}. 
For subspaces $U,V\subset\mathbb{R}^n$, the tangency constraint associated 
with $A,B\in\mathbb{S}^n$ is the affine subspace
\begin{equation*}
\mathcal{T}_{A,B}(U,V)
\triangleq
\left\{
C\in\mathbb{S}^n:\;
\begin{aligned}
& Cu = Bu,\;\; \forall\,u\in U,\\
& Cv = Av,\;\; \forall\,v\in V
\end{aligned}
\right\}
\end{equation*}
Such constraints describe how a lower bound can match $A$ and $B$ along 
selected eigendirections determined by the inertia of $A-B$, and therefore highlight directions that dominate dissipation. 
However, for a matrix-valued passivity index, the underlying set 
$\mathcal{H}$ (or $\mathcal{K}$) contains infinitely many matrices, so no single passivity matrix can satisfy tangency constraints for all frequencies.

To capture the most dissipatively critical behavior, we impose the tangency 
constraint only at the weakest passivity direction—namely, at the frequency 
where the minimum eigenvalue of $H(\omega)$ is smallest and along its 
associated eigenvector. This yields a passivity matrix that, while not 
unique, faithfully reflects the system’s least dissipative direction and 
motivates the following eigenvalue-based selection principle.

For the IFP case, we denote $\underline{H}(\omega)$ as the matrix $H(\omega)$ whose
minimum eigenvalue is the smallest among all frequencies $\omega$, i.e.,
\begin{equation*}
    \omega_\star \in \operatorname*{arg}\operatorname*{min}_{\omega\in\mathbb{R}}
    \lambda_{\min}\bigl(H(\omega)\bigr),\qquad\underline{H}(\omega) := H(\omega_\star)
\end{equation*}

Let $\underline{v}(\omega)$ be the eigenvector of $\operatorname*{min}_{\omega\in\mathbb{R}}\lambda_{\min}\bigl(H(\omega)\bigr)$, and selecting the maximal lower bound $\Phi\in\mathcal{L}_{\max}(\mathcal{H})$ satisfying the tangency-constraint $\Phi\underline{v}(\omega)=\underline{H}(\omega)\underline{v}(\omega)$, which exactly compensates the system's weakest passivity direction. 
However, such tangency constraints are generally difficult to construct and solve. Therefore, in the LMI-based design we operationalize the same structural idea by maximizing the minimum eigenvalue of the passivity matrix~$\Phi$ under the lower-bound constraints. 
Note that maximizing only the minimum eigenvalue does not guarantee a Löwner-maximal solution; one may impose a secondary objective of maximizing the trace to ensure order-theoretic maximality. 
The result IFPM and OFPM are denoted as $\Phi^{\lambda}$ and $\Xi^{\lambda}$, respectively.
It is obvious that this selection principle is consistent with the scalar passivity index when evaluated via its minimum eigenvalue, that is:
\begin{equation*}
    \phi=\lambda_{\min}(\Phi^{\lambda}),\quad \xi=\lambda_{\min}(\Xi^{\lambda})
\end{equation*}

\begin{remark}
For a generic pair of symmetric matrices, the extremal lower bounds determined by the Löwner order are aligned with eigendirections specified by the inertia of $A-B$ and are therefore not, in general, compatible with the ambient coordinate axes.  As a consequence, a fully decoupled (input–output diagonal) passivity 
matrix rarely achieves maximality in the Löwner order. Representative passivity matrices must typically incorporate coupling across different channels that are intrinsic to attaining order-theoretic optimality.
\end{remark}

\section{Passivation and Stability Analysis Based on Passivity Matrices}\label{sec-stability}
Since passivity is preserved under parallel and feedback interconnection, passivity indices are widely used to analyze the passivity and stability of interconnected subsystems in a distributed manner. When passivity indices are extended from scalars to matrices, the fundamental passivity properties of interconnected systems remain valid, while the additional directional information enables more effective control design and stability assessment. This section first characterizes how the matrix-valued passivity indices of two interconnected subsystems evolve under interconnection. We then present a passivation method based on passivity matrices, followed by stability criteria for interconnected systems derived from matrix-valued passivity indices.


\subsection{Interconnection Properties of Passivity Matrices}
Consider two subsystems $\mathcal{G}_1$ and $\mathcal{G}_2$ with the form of either a time-invariant dynamical system represented by 
\begin{equation}
    \begin{cases}
        \dot{x}_i=f_i(x_i,e_i)\\
        y_i=h_i(x_i,e_i)
    \end{cases}
    \label{eq-subsystem}
\end{equation}
or a memoryless map $y_i=h_i(t,e_i)$. Suppose each subsystem is $(\Phi_i,\Xi_i)$-passive, that is:
\begin{equation}
    e_i^{\top}y_i\ge \dot{V}_i+e_i^{\top}\Phi_ie_i+y_i^{\top}\Xi_iy_i,\quad \text{for}\quad i=1,2
    \label{eq-interconnectionV}
\end{equation}
for some storage function $V_i$.

The parallel connection and the negative feedback interconnection of these two subsystems are shown in Fig.~\ref{fig-parallel}A.
\begin{figure}[htbp]
    \centering
    \includegraphics[width=\hsize]{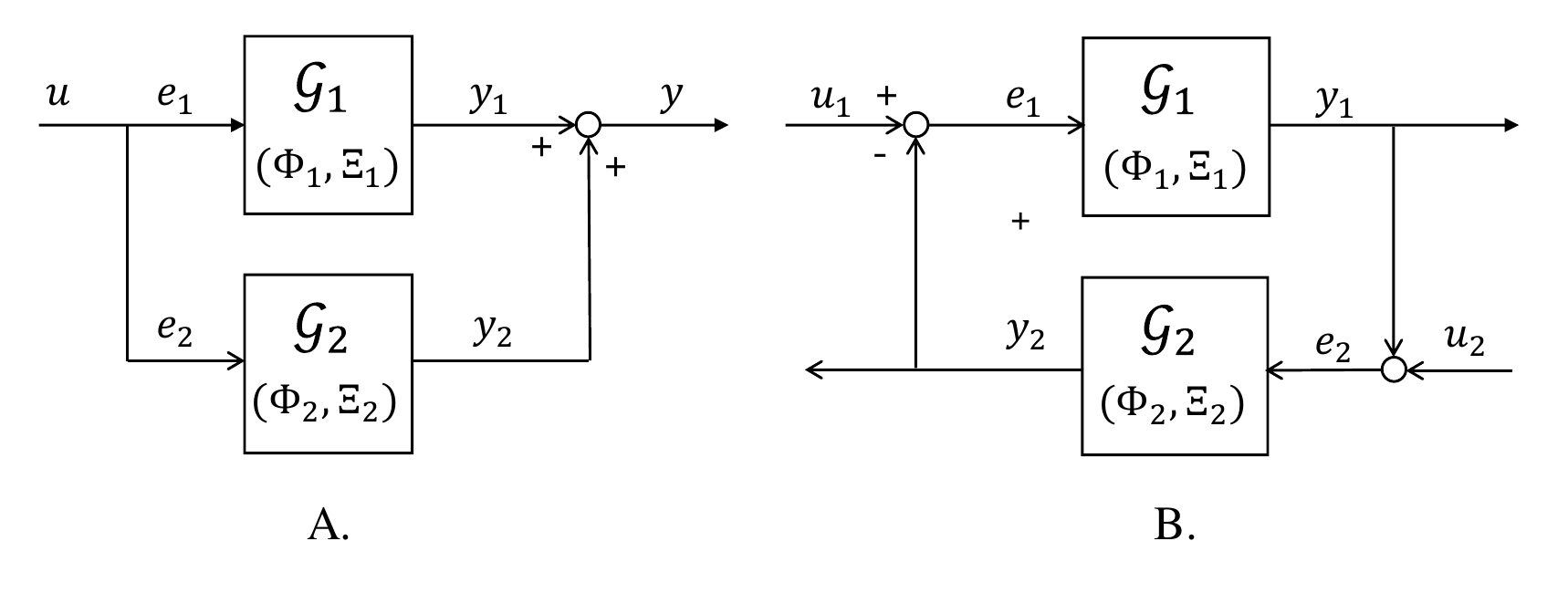}
    \caption{The block diagram of parallel connection (A) and feedback interconnection (B).}
    \label{fig-parallel}
\end{figure}

The next two theorems present systematic methods for characterizing the passivity matrices of interconnected systems.

\begin{theorem}[Parallel connection]
    Consider the parallel connection (see Fig.~\ref{fig-parallel}A) of two IF-OFP systems with the passivity matrices $(\Phi_1,\Xi_1)$ and $(\Phi_2,\Xi_2)$ respectively. If $\Xi_1$ and $\Xi_2$ are positive definite, the interconnected system with the input $u$ and the output $y$ is IF-OFP with the IFPM $\Phi=\Phi_1+\Phi_2$ and the OFPM $\Xi=(\Xi_1^{-1}+\Xi_2^{-1})^{-1}$
    \label{th-parallel connection}
\end{theorem}
\begin{proof}
    The proof can be found in Appendix~\ref{sec-app-parallel}.
\end{proof}

\begin{remark}
    If the subsystems are IFP with $\Xi_i=0$, then the parallel interconnection remains IFP with the IFPM $\Phi=\Phi_1+\Phi_2$. If the subsystems are OFP with $\Phi_i=0$, Theorem~\ref{th-parallel connection} shows that the OFP property is guaranteed for the parallel interconnection if the OFPMs $\Xi_i$ are positive definite.
\end{remark}

\begin{theorem}[Negative feedback interconnection]
Consider the negative feedback interconnection (see Fig.~\ref{fig-parallel}B) of two IF-OFP systems with the passivity matrices $(\Phi_1,\Xi_1)$ and $(\Phi_2,\Xi_2)$ respectively. Then, the overall interconnected system with
\begin{equation*}
    u=\begin{bmatrix}
        u_1^\top& u_2^\top
    \end{bmatrix}^\top\quad \text{and} \quad y=\begin{bmatrix}
        y_1^\top& y_2^\top
    \end{bmatrix}^\top
\end{equation*}
is IF--OFP with the IFPM $\Phi$ and the OFPM $\Xi$ defined by:
\[
\Phi = \mathrm{diag}(M_1,\, M_2), \Xi = \mathrm{diag}(N_1,\, N_2)
\]
where $M_1, M_2, N_1, N_2 \in \mathbb{S}^m$ satisfy
\begin{equation}
    \begin{cases}
        M_1 \prec \Phi_1,\quad
        M_2 \prec \Phi_2\\[2mm]
        N_1 \preceq \Xi_1 - \Phi_2(\Phi_2 - M_2)^{-1} M_2\\[2mm]
        N_2 \preceq \Xi_2 - \Phi_1(\Phi_1 - M_1)^{-1} M_1
    \end{cases}
    \label{eq-feedbackIDFP}
\end{equation}
\label{th-feedback interconnection}
\end{theorem}
\begin{proof}
    The proof can be found in Appendix~\ref{sec-app-feedback}.
\end{proof}

\subsection{Passivation Based on Passivity Matrices}
A major motivation for introducing matrix-valued passivity indices is that they can be used directly in passivation control design.
For input–feedforward passivation, the required compensation can be obtained straightforwardly from Theorem~\ref{th-parallel connection}, making the design relatively simple.
In contrast, output–feedback passivation is inherently more involved, since the feedback interconnection couples the passivity properties of the subsystems in a nonlinear manner.
In practical scenarios, it is often sufficient to enforce passivity only with respect to a selected pair of input and output ports rather than with respect to the full set of external ports.
The following theorem characterizes such partial passivation in the case of matrix-valued passivity indices. 
\begin{theorem}
    Consider the feedback interconnection in Fig.~\ref{fig-parallel}B. Assume $u_2=0$. The closed-loop system is passive with respect to the input $u_1$ and output $y_1$ if the passivity matrices satisfy the conditions:
    \begin{equation}
        \begin{cases}
            \Phi_2+\Xi_1\succeq 0\\
            \Xi_2\succeq 0,\quad \Phi_1\succeq 0
        \end{cases}
        \label{eq-passivation}
    \end{equation}
    Furthermore, if $\Phi_1+\Xi_2\succ 0$, the closed-loop system has the IFPM $\Phi=\Phi_2+\Xi_1$ and the OFPM $\Xi=\Xi_2(\Phi_1+\Xi_2)^{-1}\Phi_1$.
    \label{th-passivation}
\end{theorem}
\begin{proof}
    The proof can be found in Appendix~\ref{sec-app-passivation}.
\end{proof}

This reflects the complementarity between the two types of passivity: an output-feedback passivity shortage in one subsystem can be neutralized by an input-feedforward passivity surplus in the other along the shared interconnection channel. A key advantage of matrix-valued passivity indices is that they encode both passivity directions and multi-dimensional passivity intensities, thereby enabling the controller to compensate only where the system's passivity is deficient. This structural richness directly lowers the required control energy in passivation.
Consider a non-passive system $\mathcal{G}_1$ whose passivity shortage is characterized by an OFPM $\Xi_1\preceq 0$. To make the closed-loop system passive, the feedback controller $\mathcal{G}_2$ should contribute a compensating OFPM $\Xi_2\succeq 0$ and an IFPM $\Phi_2\succeq -\Xi_1$. A natural way to quantify the passivation effort is to measure the minimum control energy required to supply the necessary $\Phi_2 = -\Xi_1$. This can be expressed through the quadratic cost
\begin{equation*}
    J_{mx}(\Phi_2)=\int_0^{\infty}e_2(t)^{\top}\Phi_2e_2(t) \text{dt}=-\int_0^{\infty}e_2(t)^{\top}\Xi_1e_2(t) \text{dt}
\end{equation*}
which represents the energy injected by the controller to compensate for the passivity deficit of the system $\mathcal{G}_1$.

If only the scalar passivity index is employed, the compensation condition reduces to $\lambda_{\min}(\Phi_2)\ge-\lambda_{\min}(\Xi_1)$, which enforces an isotropic increase of dissipation in all directions, including those that do not require compensation. This yields the minimal cost $J_{sc}(\Phi_2)=-\lambda_{\min}(\Xi_1)\int e_2^{\top}e_2 \text{dt}$. In contrast, the matrix-valued formulation admits anisotropic compensation, and the resulting minimal cost $J_{mx}\le J_{sc}$, with equality occurring if and only if $\Xi_1=\xi_1I$. Hence, matrix-valued passivity indices achieve strictly lower passivation effort by avoiding unnecessary compensation in directions where the system is already sufficiently passive.

\subsection{Stability Analysis Based on Passivity Matrices}
In this subsection, we present several results on $\mathcal{L}_2$ and Lyapunov stability derived from passivity matrices.
\begin{lemma}
    If system~\eqref{eq-dynamic} is output strictly passive with an OFPM $\Xi\succ 0$, then it is finite-gain $\mathcal{L}_2$ stable and its $\mathcal{L}_2$ gain is less than or equal to $1/\lambda_{\min}(\Xi)$.
\end{lemma}
This lemma follows immediately from the standard scalar result 
\(
u^\top y \ge \dot V + \rho \left\| y \right\|^2 
\;\Longrightarrow\; 
\left\| y \right\|_{\mathcal{L}_2} \le \tfrac{1}{\rho} \left\| u \right\|_{\mathcal{L}_2},
\)
by noting that 
\(
y^\top \Xi y \ge \lambda_{\min}(\Xi)\, \left\| y \right\|^2.
\)

The $\mathcal{L}_2$ stability of the feedback interconnection is summarized in the next theorem.
\begin{theorem}
    Consider the feedback interconnection (Fig.~\ref{fig-parallel}B) of two IF-OFP systems with the passivity matrices $(\Phi_1,\Xi_1)$ and $(\Phi_2,\Xi_2)$ respectively. Then, the closed-loop map from $u$ to $y$ is finite gain $\mathcal{L}_2$ stable if
    \begin{equation}
        \Phi_1+\Xi_2 \succ 0 \quad \text{and} \quad \Phi_2+\Xi_1 \succ 0
    \end{equation}
    \label{th-L2stability}
\end{theorem}
\begin{proof}
    The proof can be found in Appendix~\ref{sec-app-L2stability}.
\end{proof}

\begin{theorem}
    Consider the feedback interconnectionof two IF-OFP time-invariant dynamical systems of the form~\eqref{eq-dynamic}, as shown in Fig.~\ref{fig-parallel}B. Suppose each feedback component is zero-state observable and has the IFPM $\Phi_i$ and the OFPM $\Xi_i$. Then the origin of the closed-loop system for $u=0$ is asymptotically stable if
    \begin{equation}
        \Phi_1+\Xi_2 \succeq 0 \quad \text{and} \quad \Phi_2+\Xi_1 \succeq 0
        \label{eq-Lyapunovstability}
    \end{equation}
    Furthermore, if $V_1+V_2$ is radially unbounded, the origin is globally asymptotically stable.
    \label{th-Lyapunovstability}
\end{theorem}
\begin{proof}
    The proof can be found in Appendix~\ref{sec-app-Lyapunovstability}.
\end{proof}

\begin{remark}
Theorem~\ref{th-Lyapunovstability}  covers the special case where $\mathcal{G}_2$ is a static input--output map $y_2 = K e_2$ and is IF--OFP with passivity matrices $(\Phi_2,\Xi_2)$. Then $\mathcal{G_2}$ can be viewed as a time-invariant dynamical system with zero dynamics and storage
function $V_2 \equiv 0$. Hence, the (global) asymptotic stability conclusion remains valid under the same matrix conditions $\Phi_1 + \Xi_2 \succeq 0$ and $\Phi_2 + \Xi_1 \succeq 0$ when $\mathcal{G_2}$ is static.
\end{remark}
These results show that extending scalar passivity indices to matrix-valued indices preserves the main classical passivity properties while substantially strengthening them. The matrix formulation maintains full compatibility with standard interconnection and stability theorems, yet provides richer directional information, enabling less conservative stability conditions and lower passivation effort. Thus, matrix-valued indices offer a natural and strictly more powerful generalization of the classical scalar framework.

\section{Illustrative Examples}\label{sec-examples}
\subsection{Feedback Passivation}
We first consider the feedback interconnection shown in Fig.~\ref{fig-parallel}B with $u_2=0$, where $\mathcal{G}_1$ is chosen as a second-order LTI system with
\begin{eqnarray*}
    \begin{aligned}
        &A=\begin{bmatrix}
            -2 & 3 \\ -8 & -10
        \end{bmatrix}, &&B=\begin{bmatrix}
            -1.3 & 3.4 \\ 3.6 & -1.7
        \end{bmatrix},\\
        &C=\begin{bmatrix}
            8 & 9 \\ 10 & 7
        \end{bmatrix}, &&D=\begin{bmatrix}
            8 & 8 \\ 6 & -8
        \end{bmatrix}
    \end{aligned}
\end{eqnarray*}
The subsystem $\mathcal{G}_1$ is output feedback passive (OFP), and its scalar OFP index is computed as $\xi_1=-0.1095$. 
To obtain the matrix-valued OFP indices, we solve the LMI~\eqref{eq-LMI} using the two optimization objectives introduced in Section~\ref{sec-orderpropwety}. 
The OFPM obtained by maximizing the trace is denoted by $\Xi_{1}^{\mathrm{tr}}$, whereas the OFPM obtained by maximizing the minimum eigenvalue is denoted by $\Xi_{1}^{\lambda}$. 
The resulting matrices are
\begin{equation*}
    \Xi_{1}^{\mathrm{tr}}
    =
    \begin{bmatrix}
        0.0373 & 0.0618\\
        0.0618 & -0.0920
    \end{bmatrix},
    \quad
    \Xi_{1}^{\lambda}
    =
    \begin{bmatrix}
        -0.06127 & 0.0176\\
        0.0176 & -0.1029
    \end{bmatrix}
\end{equation*}
Their minimum eigenvalues are $-0.1167$ and $-0.1095$, respectively. 

Consider the feedback system $\mathcal{G}_2$ as a memoryless map
\begin{equation*}
    \mathcal{G}_2: y_2=\theta Ke_2
\end{equation*}
where $K\in\mathbb{R}^{2\times 2}$ is a fixed matrix which is selected independent of the dynamics of $\mathcal{G}_1$. The scalar $\theta\ge 0$ represents a tunable feedback strength. Since $\mathcal{G}_2$ is a static mapping, it is IFP, and its associated IFPM is simply $\Phi_2=\theta K$. The corresponding scalar IFP index is therefore given by $\phi_2=\lambda_{\min}(\theta K)$. We consider three representative choices of $K$, denoted by $K_1$–$K_3$:
\begin{equation*}
    \begin{aligned}
        K_1\!=\!\begin{bmatrix} 0.987 & 0.643 \\ 0.643 & 1.013 \end{bmatrix},
        K_2\!=\!\begin{bmatrix} 0.91 & 0.149 \\ 0.149 & 1.09 \end{bmatrix},
        K_3\!=\!\begin{bmatrix} 1 & 0 \\ 0 & 1 \end{bmatrix}
    \end{aligned}
\end{equation*}

We next investigate how the closed-loop passivity evolves as the feedback gain $\theta$ increases under different choices of the static mapping $K$. For each selected $K$, $\theta$ is gradually swept from small to large values, and the exact closed-loop passivity boundary is determined from the Hermitian part of the overall closed-loop transfer function. In parallel, the certified passivity thresholds are obtained from Theorem~\ref{th-passivation} using the scalar OFP index and the two matrix-valued OFPMs—one maximizing the trace and the other maximizing the minimum eigenvalue. A smaller critical value of~$\theta$ indicates that passivity can be achieved under weaker feedback, and the closer a certified transition point is to the true boundary, the more accurately the corresponding index reflects the system’s directional passivity properties in that specific scenario. The evolution of these thresholds and their comparison across different choices of $K$ are summarized in Fig.~\ref{fig-casesbstatic}.

\begin{figure}[htbp]
    \centering
    \includegraphics[width=\hsize]{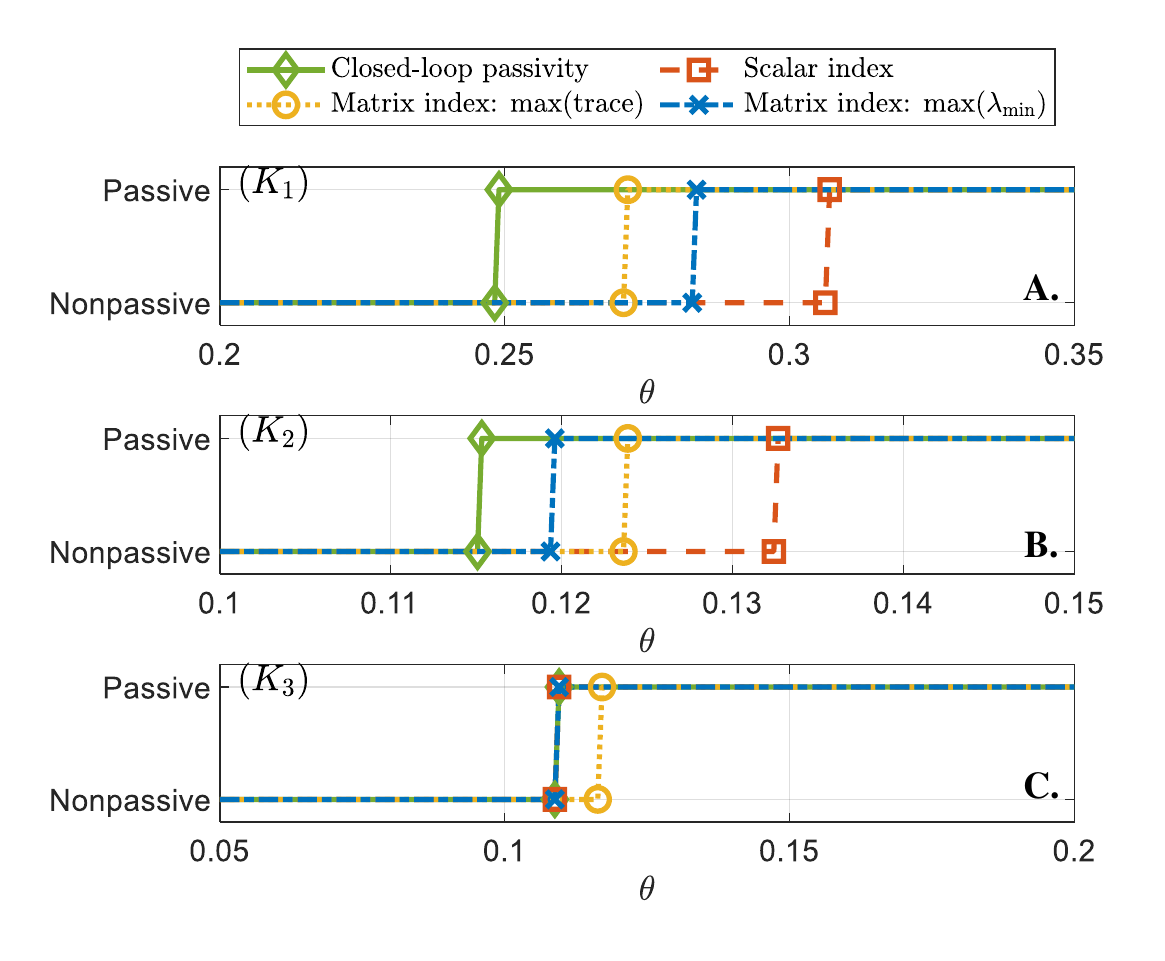}
    \caption{Passivation performance of scalar and matrix-valued passivity indices under different static feedback mappings: (A) $K_1$, (B) $K_2$, (C) $K_3$}.
    \label{fig-casesbstatic}
\end{figure}

Figure.~\ref{fig-casesbstatic}(A) shows that, for the feedback matrix $K_1$, the trace-maximizing 
passivity matrix $\Xi_1^{tr}$ yields the closest prediction to the true passivity threshold, while the minimum-eigenvalue-maximizing index $\Xi_1^{\lambda}$ is more conservative.  
In contrast, Fig.~\ref{fig-casesbstatic}(B) shows that the latter becomes more accurate for the 
feedback matrix $K_2$.  
These two cases highlight that the two matrix-valued indices may outperform one another under different anisotropic feedback structures introduced by $K$, but both consistently outperform the scalar passivity index in capturing the system's multidirectional passivity characteristics.

Figure.~\ref{fig-casesbstatic}(C) illustrates an extreme isotropic case. Here, the scalar index and the minimum-eigenvalue-maximizing index both match the true passivity boundary exactly, whereas the trace-maximizing index shows slightly inferior performance by sacrificing accuracy in the weakest passivity direction in exchange for a stronger overall index. Note that the scalar index aligns with the matrix index only in this special isotropic scenario, whereas the minimum-eigenvalue-maximizing matrix index is never worse than the scalar one.

These observations reflect a fundamental structural property of matrix-valued passivity indices: due to the partial-order structure described in Section~\ref{sec-orderpropwety}, there generally does not exist a single constant passivity matrix that simultaneously matches the system’s true passivity in all directions. Different candidate matrices capture different aspects of the system’s directional passivity.  
Therefore, if one aims to use passivity matrices for subsystem-wise decomposition and distributed analysis rather than relying on full closed-loop information, this structural conservativeness must be accepted.

\subsection{Stability Analysis}
Next, we consider a classical single-machine infinite-bus (SMIB) system, whose configuration follows the benchmark model in~\cite{li2012oscillation}. The synchronous generator is described by the standard
third–order model
\begin{equation*}
\label{eq:gen_model}
\begin{cases}
\dot{\delta} = \omega_0 \omega\\[2mm]
T_j \dot{\omega} = P_m - P_e - D \omega\\[2mm]
T_{d0}' \dot{E}_q' = E_f - E_q' - (x_d - x_d') I_d
\end{cases}
\end{equation*}
with the electrical algebraic relations $E_q' = U_q + x_d' I_d$ and $0 = U_d - x_q I_q$.
The above SMIB model is nonlinear, since the electrical variables involve $\sin{\delta}$ and $\cos{\delta}$ through the dq-axis transformation, which renders the electrical power $P_e$ a nonlinear function of the states. We analyze the passivity properties of the generator using its Hamiltonian realization~\cite{li2012oscillation}. The inputs and outputs are chosen as $u = [\,P_m\;\; E_f\,]^{\top},\qquad
y = [\,\omega_0\omega\;\; (E_q'/(x_d-x_d'))\,]^{\top}$.
One can verify that the generator is OFP with respect to the Hamiltonian
\[
H = \tfrac{1}{2} T_j \omega_0 \omega^{2}
    + \int P_e \, \mathrm{d}\delta
    + \int \frac{E_q'}{x_d - x_d'} \, \mathrm{d}E_q'
\]
and the OFPM obtained from 
$\dot{H} \le y^\top u - y^\top \Xi_{\mathrm{gen}} y$ is
\[
\Xi_{\mathrm{gen}}=\mathrm{diag}\!\left[
    {D}/{\omega_0},\;
    T_{d0}'(x_d - x_d')\right]
\]

The parameters of the SMIB generator are chosen as $\omega_0 = 2\pi\times 50,\quad T_j = 15,\quad D = 8,\quad T_{d0}' = 5, x_d = 0.5,\quad x_q = 0.5,\quad x_d' = 0.35,\quad U = 1.0$
and the operating point is $P_m = 0.8, E_f = 1.2$. Consider a static output feedback
\[
K=\begin{bmatrix} K_{11} & 0.1 \\[1mm] 0.1 & K_{22} \end{bmatrix}
\]
which is interconnected with the generator in a negative-feedback configuration.
By varying the values of $K_{11}$ and $K_{22}$, we evaluate the stability of the overall interconnection under both the scalar and matrix passivity indices. Using Theorem~\ref{th-Lyapunovstability}, Fig.~\ref{fig-genregion} compares the estimated stability regions in the $(K_{11}-K_{22})$ plane certified by the scalar passivity indices (Fig.~\ref{fig-genregion}A) and the matrix-valued passivity indices (Fig.~\ref{fig-genregion}B). The red areas indicate the controller gains that can be certified as stable by the corresponding passivity condition, the blue areas denote the unstable region identified by eigenvalue analysis, and the black curve represents the small-signal stability boundary defined by $\max \Re\{\lambda(A)\} = 0$. The matrix-valued index certifies a much larger admissible region than the scalar one, demonstrating a less conservative assessment; Cases 1–4 mark the gains used for time-domain validation.


\begin{figure}[htbp]
    \centering
    \includegraphics[width=\hsize]{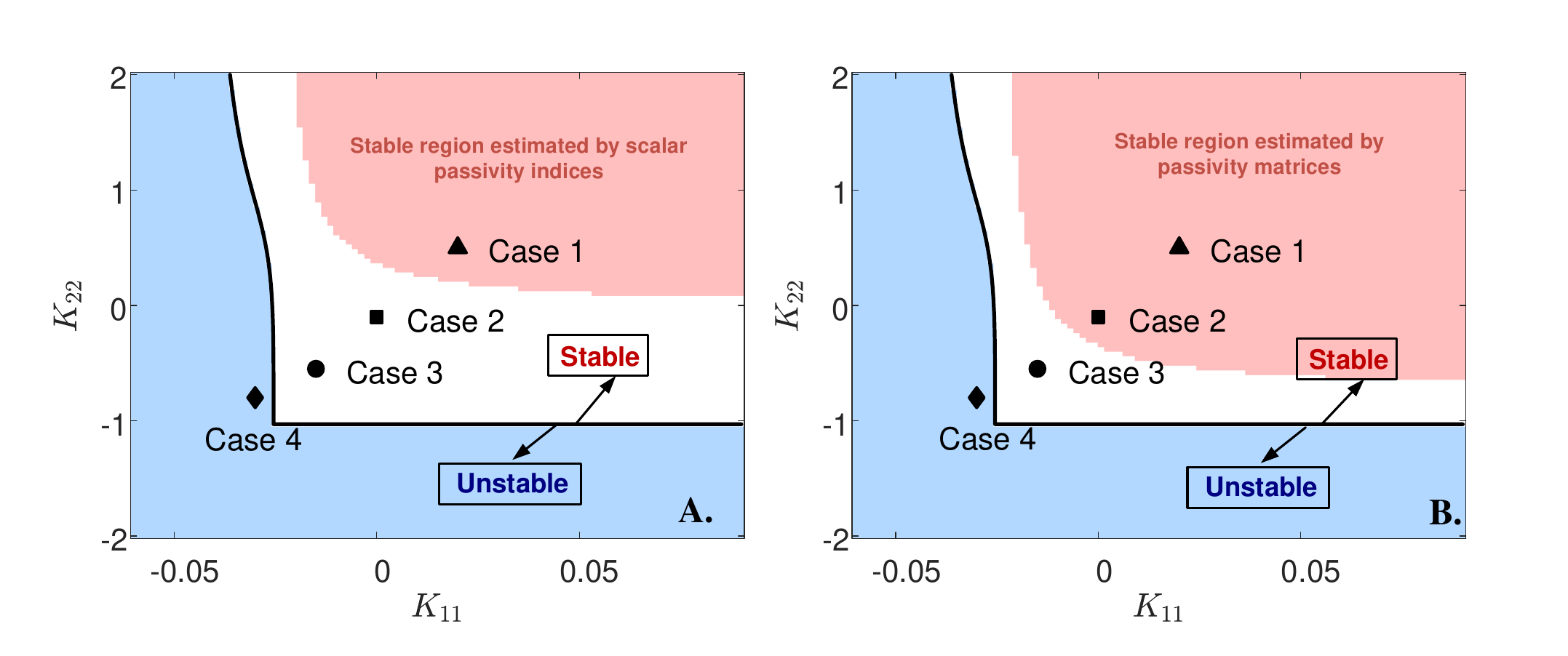}
    \caption{Stability regions of the SMIB system identified by the scalar passivity indices (A) and the matrix-valued passivity indices (B).}
    \label{fig-genregion}
\end{figure}

To further validate the results, four representative controller gains $K$ are selected for time-domain simulations: (Case~1) both indices certify stability, (Case~2) only the matrix index certifies stability, (Case~3) neither index can certify stability, and (Case~4) the linearized eigenvalue test indicates small-signal instability at the corresponding equilibrium. For each case, the closed-loop equilibrium $x^*(K)$ is first computed, and the simulation is initialized at a fixed-distance perturbation $x_0=x^*(K)+rd$ with the same radius $r=0.03$ and direction $d$ for all cases, thus all cases are tested under an identical small-signal disturbance level around their own equilibria. The system is then simulated over the same time window, and the transient responses $\delta(t)$, $\omega(t)$, and $P_e(t)$ are reported in Fig.~\ref{fig-simulation}. The corresponding transient responses show that Cases 1–3 are stable, with different damping levels, while Case 4 exhibits divergent oscillations and becomes unstable. These time-domain responses provide additional evidence that the stability region certified by the passivity matrix is less conservative than that of the scalar index.


\begin{figure}[htbp]
    \centering
    \includegraphics[width=\hsize]{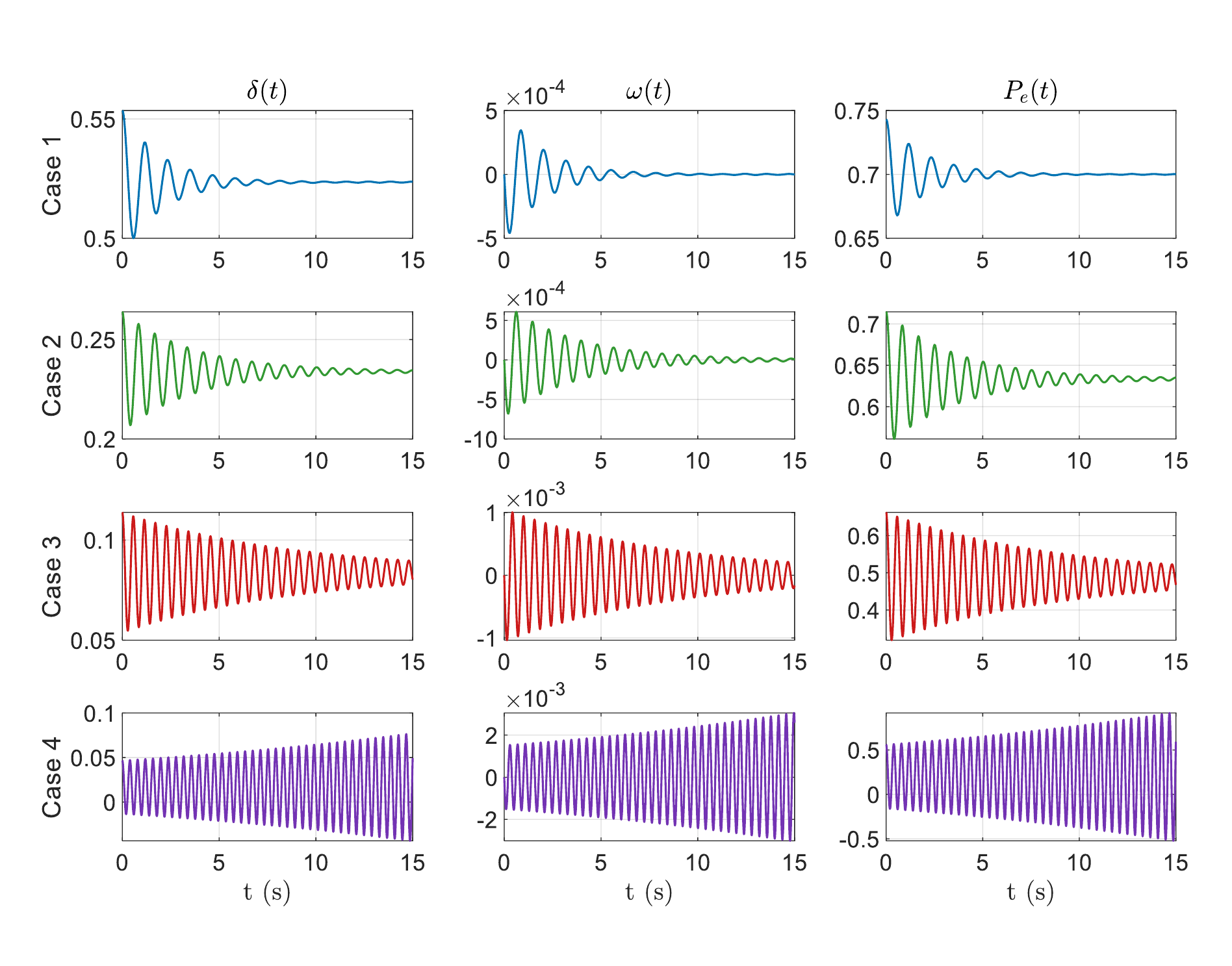}
    \caption{Time-domain simulation of the SMIB system under the four representative controller gains (Case 1–4).}
    \label{fig-simulation}
\end{figure}

\section{Conclusion}\label{sec-conclusion}
In this paper, we have established a systematic framework of matrix-valued passivity indices by extending classical scalar indices. The passivity matrices can better accommodate the directional complexity of MIMO systems. By interpreting passivity matrices in terms of the curvature of the dissipative functional, we have provided a rigorous geometric foundation that reveals both the intensity and directionality of energy dissipation. This finding sheds new light on understanding the intrinsic energy structure of passivizable systems.
For linear time-invariant systems, the application of the Loewner partial order, combined with proposed selection principles such as trace maximization and minimum eigenvalue maximization, has enabled a tractable computation of these indices via LMIs.

Theoretical analysis and numerical validations demonstrate that the proposed matrix-valued indices can capture cross-channel coupling, thereby allowing for less conservative stability assessments and reducing the necessary control effort in passivation tasks by avoiding isotropic over-compensation.

\renewcommand{\baselinestretch}{0.92}
\bibliographystyle{IEEEtran}
\bibliography{reference.bib}
\renewcommand{\baselinestretch}{0.93}


\appendices

\renewcommand{\thedefinition}{\thesection.\arabic{definition}}
\renewcommand{\theassumption}{\thesection.\arabic{assumption}}
\renewcommand{\thelemma}{\thesection.\arabic{lemma}}
\renewcommand{\theremark}{\thesection.\arabic{remark}}
\renewcommand{\thetheorem}{\thesection.\arabic{theorem}}
\renewcommand{\thecondition}{\thesection.\arabic{condition}}
\renewcommand{\theproposition}{\thesection.\arabic{proposition}}

\setcounter{definition}{0}
\setcounter{assumption}{0}
\setcounter{lemma}{0}
\setcounter{remark}{0}
\setcounter{theorem}{0}
\setcounter{condition}{0}
\setcounter{proposition}{0}

\counterwithin{equation}{section}
\section{Proof of Lemma~\ref{le-order}}\label{sec-app-orderstructure}
\begin{proof}
    The OFP of system~\eqref{eq-dynamic} indicates that
    \begin{equation}
        \dot{V}\le u^{\top}y-y^{\top}\rho(y), \forall(x,u)\in\mathcal{X}\times\mathcal{U}
        \label{eq-app-order1}
    \end{equation}
    Since $\rho(y)$ belongs to the sector $[\Xi,\infty]$, we have $y^\top[\rho(y)-\Xi y]\ge 0,\quad\forall(x,u)\in\mathcal{X}\times\mathcal{U}$.
    Substituting it into~\eqref{eq-app-order1} yields
    \begin{equation*}
        \dot{V}\le u^{\top}y- y^{\top}\Xi y,\quad\forall(x,u)\in\mathcal{X}\times\mathcal{U}
    \end{equation*}
    The proof for the IFP case is analogous and omitted here.
\end{proof}

\section{Proof of Theorem~\ref{th-parallel connection}}\label{sec-app-parallel}
\begin{proof}
    Let $V=V_1+V_2$. According to~\eqref{eq-interconnectionV},  the interconnection relative $u=e_1=e_2$ and $y=y_1+y_2$ yields
    \begin{equation*}
        \begin{aligned}
            \dot{V}=&\dot{V}_1+\dot{V}_2 \\
            \le& e_1^{\top} y_1 - e_1^{\top} \Phi_1 e_1-y_1^{\top}\Xi_1y_1+e_2^{\top} y_2 - e_2^{\top} \Phi_2 e_2-y_2^{\top}\Xi_2y_2\\
            =&u^{\top}(y_1+y_2)-u^{\top}(\Phi_1+\Phi_2)u-\begin{bmatrix}
                y_1\\y_2
            \end{bmatrix}^{\top}\begin{bmatrix}
                \Xi_1&0\\0&\Xi_2
            \end{bmatrix}\begin{bmatrix}
                y_1\\y_2
            \end{bmatrix}
        \end{aligned}
    \end{equation*}
    Denote $N=(\Xi_1^{-1}+\Xi_2^{-1})^{-1}$, the quadratic form of $[y_1,y_2]$ could be rewritten as
    \begin{equation*}
        \begin{aligned}
            &-\begin{bmatrix}
                y_1\\y_2
            \end{bmatrix}^{\top}\begin{bmatrix}
                \Xi_1&0\\0&\Xi_2
            \end{bmatrix}\begin{bmatrix}
                y_1\\y_2
            \end{bmatrix}
            =-y^{\top}Ny-\begin{bmatrix}
                y_1\\y_2
            \end{bmatrix}^{\top}M\begin{bmatrix}
                y_1\\y_2
            \end{bmatrix}
        \end{aligned}
    \end{equation*}
    where
    \[
        M :=
        \begin{bmatrix}
            \Xi_1-N & -N\\[2pt]
            -N      & \Xi_2-N
        \end{bmatrix}
    \]
    It remains to show that the block matrix $M$ is positive semidefinite. Note that $\Xi_1\succ 0$ and $\Xi_2\succ 0$, so we have $\Xi_1^{-1}\succ 0$, $\Xi_2^{-1}\succ 0$, hence $\Xi_1^{-1}\prec \Xi_1^{-1}+\Xi_2^{-1}$. This is equivalent to:
    \begin{equation}
        \Xi_1\succ (\Xi_1^{-1}+\Xi_2^{-1})^{-1}=N
        \label{eq-proof1.1}
    \end{equation}
    We also have completely symmetrical results for $\Xi_2$: $\Xi_2\succ N$.
    The Schur complement of $\Xi_1-N$ is:
    \begin{equation}
        \begin{aligned}
            &\Xi_2-N-N(\Xi_1-N)^{-1}N\\
            =&\Xi_2-N(\Xi_1-N)^{-1}(\Xi_1-N+N)\\
            =&\Xi_2-(\Xi_1N^{-1}-I)^{-1}\Xi_1
            =\Xi_2-(N^{-1}-\Xi_1^{-1})^{-1}\\
            =&\Xi_2-\Xi_2=0
        \end{aligned}
        \label{eq-proof1.2}
    \end{equation}
    According to~\eqref{eq-proof1.1} and~\eqref{eq-proof1.2}, we have $M\succeq 0$, thus
    \begin{equation*}
        \dot{V}\le u^{\top}y-u^{\top}(\Phi_1+\Phi_2)u-y^{\top}(\Xi_1^{-1}+\Xi_2^{-1})^{-1}y
    \end{equation*}
    which completes the proof.
\end{proof}

\section{Proof of Theorem~\ref{th-feedback interconnection}}\label{sec-app-feedback}
\begin{proof}
    Let $V=V_1+V_2$. According to~\eqref{eq-interconnectionV} and use the interconnection relations $e_1=u_1-y_2$ and $e_2=u_2+y_1$, we have
    \begin{equation}
        \begin{aligned}
            \dot{V}=&\dot{V}_1+\dot{V}_2 \\
            \le& e_1^{\top} y_1 - e_1^{\top} \Phi_1 e_1-y_1^{\top}\Xi_1y_1+e_2^{\top} y_2 - e_2^{\top} \Phi_2 e_2-y_2^{\top}\Xi_2y_2\\
            =&(u_1-y_2)^{\top}y_1+(u_2+y_1)^{\top}y_2-(u_1-y_2)^{\top}\Phi_1(u_1-y_2)\\
            &-(u_2+y_1)^{\top}\Phi_2(u_2+y_1)-y_1^{\top}\Xi_1y_1-y_2^{\top}\Xi_2y_2\\
            =&u^{\top}y-u^{\top}\begin{bmatrix}
                M_1 & 0\\0 & M_2
            \end{bmatrix}u-y^{\top}\begin{bmatrix}
                N_1 & 0\\0 & N_2
            \end{bmatrix}y\\
            &-\begin{bmatrix}
                u_1\\y_2
            \end{bmatrix}^{\top}E\begin{bmatrix}
                u_1\\y_2
            \end{bmatrix}-\begin{bmatrix}
                u_2\\y_1
            \end{bmatrix}^{\top}F\begin{bmatrix}
                u_2\\y_1
            \end{bmatrix}
        \end{aligned}
        \label{eq-prooffb1}
    \end{equation}
    where
    \begin{equation*}
    \begin{aligned}
        E&=\begin{bmatrix}
                \Phi_1-M_1 & -\Phi_1\\-\Phi_1 & \Xi_2+\Phi_1-N_2
            \end{bmatrix},\\
        F&=\begin{bmatrix}
                \Phi_2-M_2 & \Phi_2\\ \Phi_2 & \Xi_1+\Phi_2-N_1
            \end{bmatrix}
    \end{aligned}
    \end{equation*}
    Since $M_1, M_2, N_1, N_2 \in \mathbb{S}^m$ satisfy~\eqref{eq-feedbackIDFP}, we have:
    \begin{equation*}
    \begin{aligned}
        N_1&\preceq\Xi_1+\Phi_2(\Phi_2-M_2)^{-1}(\Phi_2-M_2-\Phi_2)\\
        &=\Xi_1+\Phi_2-\Phi_2(\Phi_2-M_2)^{-1}\Phi_2\\
        N_2&\preceq\Xi_2+\Phi_1(\Phi_1 - M_1)^{-1}(\Phi_1-M_1-\Phi_1)\\
        &=\Xi_2+\Phi_1-\Phi_1(\Phi_1 - M_1)^{-1}\Phi_1\\
    \end{aligned}
    \end{equation*}
    The above two inequalities together with $\Phi_1-M_1\succ 0$ and $\Phi_2-M_2\succ 0$ imply, by the Schur complement lemma, that $E\succeq 0$ and $F\succeq 0$, thus
    \begin{equation*}
        \dot{V}=u^{\top}y-u^{\top}\begin{bmatrix}
                M_1 & 0\\0 & M_2
            \end{bmatrix}u-y^{\top}\begin{bmatrix}
                N_1 & 0\\0 & N_2
            \end{bmatrix}y
    \end{equation*}
    which completes the proof.
\end{proof}

\section{Proof of Theorem~\ref{th-passivation}}\label{sec-app-passivation}
\begin{proof}
    Setting $u_2=0$ in~\eqref{eq-prooffb1} yields
    \begin{equation}
        \begin{aligned}
            \dot{V}=&\dot{V}_1+\dot{V}_2 \\
            \le&(u_1-y_2)^{\top}y_1+y_1^{\top}y_2-e_1^{\top}\Phi_1e_1\\
            &-y_1^{\top}\Phi_2y_1-y_1^{\top}\Xi_1y_1-(u_1-e_2)^{\top}\Xi_2(u_1-e_2)\\
            =&u_1^{\top}y_1-y_1^{\top}(\Xi_1+\Phi_2)y_1-\begin{bmatrix}
                e_1\\u_1
            \end{bmatrix}^{\top}M\begin{bmatrix}
                e_1\\u_1
            \end{bmatrix}
        \end{aligned}
        \label{eq-proofpassivation1}
    \end{equation}
    where
    \begin{equation*}
        M=\begin{bmatrix}
            \Phi_1+\Xi_2 & -\Xi_2\\-\Xi_2 & \Xi_2
        \end{bmatrix}
    \end{equation*}
    Since the passivity matrices satisfy~\eqref{eq-passivation}, we have $\Xi_2\succeq 0$ and $\Phi_1+\Xi_2-\Xi_2\Xi_2^{-1}\Xi_2=\Phi_1\succeq 0$, thus $M\succeq 0$, $\dot{V}\le u_1^{\top}y_1$, which indicates that the closed loop system is passive with $u_1$ and $y_1$.
    In addition,~\eqref{eq-proofpassivation1} could be arranged as
    \begin{equation*}
        \begin{aligned}
            \dot{V}\le& u_1^{\top}y_1-u_1^{\top}\left[\Xi_2(\Phi_1+\Xi_2)^{-1}\Phi_1\right] u_1\\
            &-y_1^{\top}(\Xi_1+\Phi_2)y_1-\begin{bmatrix}
                e_1\\u_1
            \end{bmatrix}^{\top}M'\begin{bmatrix}
                e_1\\u_1
            \end{bmatrix}
        \end{aligned}
    \end{equation*}
    where
    \begin{equation*}
        M'=\begin{bmatrix}
            \Phi_1+\Xi_2 & -\Xi_2\\-\Xi_2 & \Xi_2-\Xi_2(\Phi_1+\Xi_2)^{-1}\Phi_1
        \end{bmatrix}
    \end{equation*}
    Since $\Phi_1+\Xi_2\succ 0$ and the Schur complement of $M'$ with $\Phi_1+\Xi_2$ is
    \begin{equation*}
        \begin{aligned}
            &\Xi_2-\Xi_2(\Phi_1+\Xi_2)^{-1}\Phi_1-\Xi_2(\Phi_1+\Xi_2)^{-1}\Xi_2\\
            =&\Xi_2-\Xi_2(\Phi_1+\Xi_2)^{-1}(\Phi_1+\Xi_2)=0
        \end{aligned}
    \end{equation*}
    Then, $M'\succeq 0$, and~\eqref{eq-proofpassivation1} becomes
    \begin{equation*}
        \dot{V}\le u_1^{\top}y_1-\left[\Xi_2(\Phi_1+\Xi_2)^{-1}\Phi_1\right] u_1-y_1^{\top}(\Xi_1+\Phi_2)y_1
    \end{equation*}
\end{proof}

\section{Proof of Theorem~\ref{th-L2stability}}\label{sec-app-L2stability}
\begin{proof}
    Rearranging the terms in~\eqref{eq-prooffb1} yields
    \begin{equation*}
        \begin{aligned}
            \dot{V}\le u^{\top}Ny-u^{\top}Mu-y^{\top}Ly
        \end{aligned}
    \end{equation*}
    where
    \begin{equation*}
    \begin{aligned}
        N&=\begin{bmatrix}
            I & 2\Phi_1\\-2\Phi_2 & I
        \end{bmatrix},M=\begin{bmatrix}
            \Phi_1 &0\\0& -\Phi_2
        \end{bmatrix},\\
        L&=\begin{bmatrix}
            \Xi_1+\Phi_2&0\\0&\Xi_2+\Phi_1
        \end{bmatrix}
    \end{aligned}
    \end{equation*}
    Let $a=\lambda_{\min}(L)$, $b=\|N\|_2\ge 0$ and $c=\|M\|_2\ge 0$. Then
    \begin{equation*}
    \begin{aligned}
        \dot V
        &\le -a\|y\|_2^2 + b\|u\|_2\|y\|_2 + c\|u\|_2^2 \\
        &= -\frac{1}{2a}\bigl(b\|u\|_2 - a\|y\|_2\bigr)^2
           + \frac{b^2}{2a}\|u\|_2^2 - \frac{a}{2}\|y\|_2^2
           + c\|u\|_2^2 \\
        &\le \frac{b^2+2ac}{2a}\|u\|_2^2 - \frac{a}{2}\|y\|_2^2 .
    \end{aligned}
    \end{equation*}
    Integrating over $[0,\tau]$, using $V(x)\ge 0$, and taking the square roots, we arrive at
    \[
    \|y_\tau\|_{\mathcal{L}_2}
    \le \sqrt{\frac{b^2+2ac}{a}}\,\|u_\tau\|_{\mathcal{L}_2}
    + \sqrt{\frac{2V(x(0))}{a}},
    \]
    which completes the proof.
\end{proof}

\section{Proof of Theorem~\ref{th-Lyapunovstability}}\label{sec-app-Lyapunovstability}
\begin{proof}
    Select $V=V_1+V_2$ as the Lyapunov function candidate. Setting $u=0$ in~\eqref{eq-prooffb1} we have
    \begin{equation*}
    \begin{aligned}
        \dot{V}\le&-y_2^{\top}\Phi_1y_2-y_1^{\top}\Phi_2y_1-y_1^{\top}\Xi_1y_1-y_2^{\top}\Xi_2y_2\\
        =&-y_1^{\top}(\Phi_2+\Xi_1)y_1-y_2^{\top}(\Phi_1+\Xi_2)y_2
    \end{aligned}
    \end{equation*}
    The condition~\eqref{eq-Lyapunovstability} shows that $\dot{V}\le 0$ and $\dot{V}=0\Rightarrow y_1=0,y_2=0$. The zero-state observability of the subsystems implies that $\dot{V}=0\Rightarrow x=0$. The conclusion follows from the invariance principle.
\end{proof}

\end{document}